\documentclass[aps,pra,reprint,showpacs,onecolumn,notitlepage,superscriptaddress]{revtex4-1}

\usepackage[latin1]{inputenc}

\usepackage[x11names]{xcolor}
\usepackage{amsmath}
\usepackage{amsthm}
\usepackage{amssymb}
\usepackage{graphicx}
\usepackage{dsfont}
\usepackage{placeins}
\usepackage{subfigure}

\newtheorem{theorem}{Theorem}

\newcommand{\1}{\mathds{1}}

\newcommand{\ket}[1]{| #1 \rangle}
\newcommand{\bra}[1]{\langle #1 |}

\newcommand{\proj}[1]{\left| #1\right>\left< #1\right|}

\DeclareRobustCommand{\em}{\it}

\DeclareMathOperator{\tr}{tr}

\begin{document}
\title{Multi-partite entanglement speeds up quantum key distribution in networks}
\author{Michael Epping}
\email{epping@hhu.de}
\affiliation{Institute for Quantum Computing, University of Waterloo, 200 University Ave. West, N2L 3G1 Waterloo, Ontario, Canada}
\affiliation{Institut f\"{u}r Theoretische Physik III, Heinrich-Heine-Universit\"{a}t D\"{u}sseldorf, Universit\"{a}tsstr. 1, D-40225
	D\"{u}sseldorf, Germany}
\author{Hermann Kampermann}
\affiliation{Institut f\"{u}r Theoretische Physik III, Heinrich-Heine-Universit\"{a}t D\"{u}sseldorf, Universit\"{a}tsstr. 1, D-40225
	D\"{u}sseldorf, Germany}
\author{Chiara Macchiavello}
\affiliation{Dipartimento di Fisica, Universit\`a di Pavia, and INFN-Sezione di Pavia, via Bassi 6, 27100 Pavia, Italy}
\author{Dagmar Bru\ss}
\affiliation{Institut f\"{u}r Theoretische Physik III, Heinrich-Heine-Universit\"{a}t D\"{u}sseldorf, Universit\"{a}tsstr. 1, D-40225
D\"{u}sseldorf, Germany}
\pacs{03.67.Dd,03.67.Bg,03.67.Pp} % Quantum Cryptography, Entanglement production, Quantum error correction

\begin{abstract}
The laws of quantum mechanics allow for the distribution of
a secret random key between two parties.
Here we analyse the security of a protocol for
establishing a common secret key between N parties (i.e. a conference key),
using resource states with genuine N-partite entanglement.
We compare this protocol to conference key distribution via
bipartite entanglement, regarding the required resources,
 achievable secret key rates and threshold qubit error rates. Furthermore we discuss quantum networks with bottlenecks for which our multipartite entanglement-based protocol can benefit from network coding, while the bipartite protocol cannot. It is shown how this advantage leads to a higher secret key rate.
% and that the performance difference increases with the number of parties.
\end{abstract}

\maketitle
\noindent

%\section{Introduction}
In the quantum world, randomness and security
are built-in properties~\cite{G+02,DLH06,B+07}: two parties may establish
a random secret key by exploiting
the no-cloning theorem~\cite{WZ82}, as in the  BB84 protocol \cite{BB84}, or by
using the monogamy of entanglement~\cite{CKW00}, as in the Ekert protocol~\cite{Ekert91}.
Several variations of these seminal
protocols have been suggested~\cite{Bruss98,GG02,ZSW08,LCQ12,VV14}, and  their security has been analysed
in detail~\cite{LC99,SP00,May01,Renner05,Sca+09,TL15,CML16}.\\
In the advent of quantum technologies, much
effort is devoted to building
quantum networks~\cite{ACL07,SECOQC09,LOW10,Sas+11,Met+13,Sat+16} and
creating global quantum states across them~\cite{EKB16,EKB16b}. Thus, the generalization of quantum key distribution (QKD) to  multipartite scenarios
is topical. In order to establish a common secret key (the {\em conference key})
for $N$ parties, one
 can follow mainly two different paths: building up the multipartite key
 from bipartite QKD links (2QKD)~\cite{SS03}, see Fig.~\ref{fig:bipartite},
or exploiting correlations of
genuinely multipartite entangled states (NQKD)~\cite{Cabello00,SG01,CL05,Fu15},
see Fig.~\ref{fig:multipartite}. \\
% Zhu2015 (only bipartite key)
\begin{figure}[tp]%
{ }\hfill\subfigure[Establishing a conference key with bipartite entanglement (2QKD).\label{fig:bipartite}]{\includegraphics[scale=1.0]{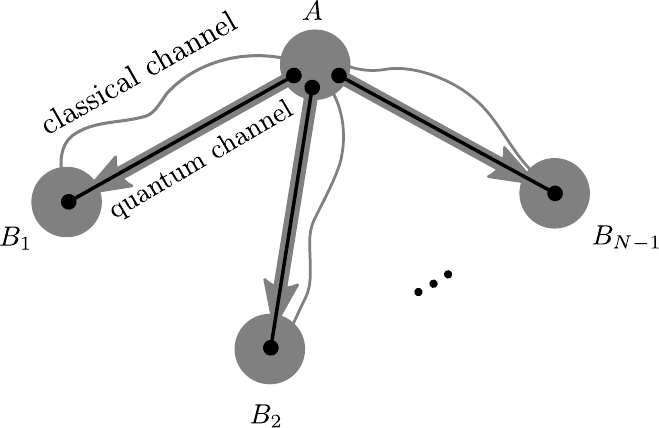}} \hfill %
\subfigure[Establishing a conference key with multipartite entanglement (NQKD).\label{fig:multipartite}]{\includegraphics[scale=1.0]{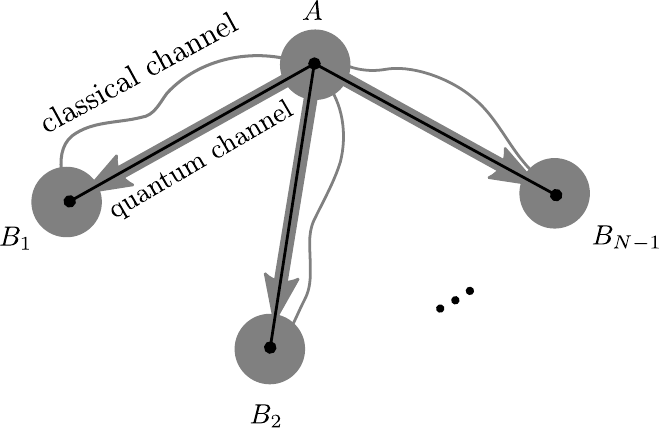}} \hfill { }%
\caption{The setup for $N$-partite conference key distribution. Black disks are qubits and the black lines connecting them indicate entanglement. Here, all quantum states are produced by Alice (A), who sends one subsystem to each of the other parties $B_i$. Both protocols require additional classical communication which is sent via open but authenticated channels. The grey background indicates the network infrastructure, i.e. the channels and nodes.
}\label{fig:schemes}%
\end{figure}%
In this article we provide an information theoretic security analysis of NQKD, by
generalising methods developed for 2QKD in~\cite{Renner05,RGK05},
and perform an analytical calculation of secret key rates.
This enables us to quantitatively compare the two approaches;
we find that NQKD may outperform 2QKD, for example in networks with bottlenecks.

The article is structured as follows. In the Section~\ref{sec:NQKD} we introduce the NQKD protocol and its prepare-and-measure variant, perform a detailed security analysis and the secret key rate calculation. In Section~\ref{sec:implementation} we define the 2QKD protocol,  summarise the steps of the NQKD protocol in an implementation and calculate the secret key rate for the example of a depolarised state. We explicitly model noise introduced by imperfect gates and channels in order to compare the performance of the two different approaches. Quantum networks are discussed in Section~\ref{sec:networks}. The article concludes with a discussion of the results.
\section{Multipartite QKD: protocol and security analysis}\label{sec:NQKD}
The entanglement-based Ekert protocol~\cite{Ekert91} can be generalised
to $N$ parties as follows, see also \cite{CL05}.  The parties
$A$ and $B_1$, $B_2$, ..., $B_{N-1}$  share an $N$-partite entangled state
and perform local projective measurements.
The best performance in the ideal (noiseless) case is ensured
if one requires that the measurement outcomes of all parties are perfectly correlated
for one set of local bases -- which we can choose
without loss of generality to be the
$Z$-bases --
and occur with a uniform distribution.
 The only pure  $N$-qubit quantum state that
fulfils these requirements is the Greenberger-Horne-Zeilinger (GHZ) state~\cite{GHZ07};
however, for  $N\geq3$, the existence of perfect correlations in one set of
 bases forbids perfect
correlations (even only pairwise) in any other bases,
see Appendix~\ref{app:GHZ}. We remark that other protocols with less
than perfectly correlated resource states are possible, but will
introduce intrinsic errors~\cite{ZXP15}.\\

\subsection{The protocol for N-party quantum conference key distribution} 
The protocol for N-party quantum conference key distribution (NQKD), with $N\geq 2$, consists of the following basic steps:
\begin{itemize}
\item[1)] {\sl State preparation:} The parties $A$ and $B_i$, ${i=1,2,...,N-1}$, share the $N$-qubit GHZ state
\begin{equation}
\ket{GHZ} = \frac{1}{\sqrt{2}}\left(\ket{0}^{\otimes N} + \ket{1}^{\otimes N}\right).\label{eq:GHZ}
\end{equation}
\item[2)] {\sl Measurement:}
	There are two types of measurements. First type:
Party $A$ and parties $B_i$ measure their respective qubits in the $Z$-basis. Second type: They measure randomly, with equal probability, in the $X$- or $Y$-basis. Similar to the standard bipartite QKD protocol~\cite{Lo2005}, the latter case is much less frequent. The parties know the type of the measurement from a short pre-shared random key.

\item[3)] {\sl Parameter estimation:}
	The parties announce the measurement bases and outcomes for the second type and an equal number of randomly chosen rounds of the first type. The announced data allows to estimate the parameters $Q_X$ and $Q_Z$, which determine the secret key rate, see below.
\item[4)] {\sl Classical post-processing:} As
 in the bipartite protocol,
error correction and privacy amplification is performed, for the
details see below.
\end{itemize}
Note that the state preparation in step 1) can be achieved by locally preparing the GHZ-state at Alice's site and sending qubits to the Bobs (see Fig.~\ref{fig:multipartite}), or any suitable sub-protocol that achieves the same task. We analyse the distribution via quantum repeaters~\cite{EKB16} and quantum
network coding~\cite{EKB16b} below.\\
In the following section we briefly discuss prepare-and-measure variants of conference key distribution. Because the security proof of the NQKD protocol is done in the entanglement-based picture, this description is not necessary for understanding the rest of the paper.\\

\subsection{N-party prepare-and-measure schemes}
We now sketch two different prepare-and-measure schemes for conference key distribution.\\
{\sl 1) Preparing and measuring single qubits:} Single qubits are experimentally easier to prepare and
to distribute than entangled states. Thus, establishing a
conference key for $N$ parties by using single qubits is
an interesting possibility, which has been studied for the
case $N=3$ in  \cite{Ryu07}.\\
The protocol proceeds in complete analogy to the case of
$N=2$, e.g. for BB84~\cite{BB84}: Alice prepares randomly $N-1$ copies of
a state $\ket{\phi_k}, k=1,...4,$ taken from the set
$S_{BB84}=\{\ket{0},\ket{1},\ket{+},\ket{-}\}$,
with $\ket{\pm}=1/\sqrt{2}(\ket{0}\pm\ket{1})$. She sends each
party $B_i$, with $i=1,...,N$, one of the copies.  Each party
$B_i$ measures in the $Z$- or $X$-basis.
In the sifting step, the $N$ parties keep only those cases
where all parties used the same basis, and thus
establish a joint key.\\
We point out, however, that the secret key rate
in this scenario decreases
with increasing $N$, even for perfect channels and measurements,
and goes to zero for $N\rightarrow\infty$: an eavesdropper
can eavesdrop on all $N-1$ sent states at the same time,
i.e. she has to distinguish the four global states
$ |\phi_k>^{\otimes (N-1)}$, pairs of which have either overlap 0 or
$(1/\sqrt{2})^{N-1}$, i.e. the distinguishability increases with increasing
$N$. In the limit of infinite $N$ the four global states are
orthogonal and therefore perfectly distinguishable.

Thus, this prepare-and-measure-scheme is (for $N\geq 3$) {\em not} equivalent
to entanglement-based NQKD as described in the present article.

{\sl 2) Prepare-and-measure equivalent of NQKD:} The entanglement-based protocol NQKD described above can be
formulated
as a prepare-and-measure protocol, analogous to the six-state protocol~\cite{Bruss98}. Instead of producing the GHZ state of Eq.~(\ref{eq:GHZ}) and measuring her qubit afterwards, Alice can directly produce the ($N-1$)-qubit projection of the GHZ state according to her fictitious, random outcome.
Thus, for the $X$-, $Y$- and $Z$-basis, the six different $(N-1)$-qubit states she distributes among the Bobs, are
\begin{equation}
\begin{aligned}
\ket{\psi_{x,\pm}}=&\frac{1}{\sqrt{2}}(\ket{00...0}\pm \ket{11...1}),\\
\ket{\psi_{y,\pm}}=&\frac{1}{\sqrt{2}}(\ket{00...0}\mp i\ket{11...1}),\\
\ket{\psi_{z,0}}=&\ket{00...0} \text{ and }\ket{\psi_{z,1}}=\ket{11...1}.
\end{aligned}
\end{equation}
This protocol is equivalent to NQKD,
because it reproduces the correlations between $A$, $B_1$, ... $B_{N-1}$. Note that the six-state protocol is included as the special case $N=2$. The described protocol uses $(N-1)$-partite entanglement for four of the sent states, which are, however, sent much less frequent than the two product states. This fact renders an experimental implementation of our protocol more realistic than the entanglement-based description might suggest.\\
In the remainder of this article we use the equivalent entanglement-based description of the NQKD protocol.\\

\subsection{Security analysis of the N-party quantum key distribution} 
The composable security definition of the bipartite
scenario~\cite{Renner05,Sca+09} can be generalised
in an analogous way to  the $N$-partite case.
Our security analysis proceeds along analogous lines as the
bipartite case in \cite{RGK05}. See Appendix~\ref{app:security} for explicit details of these generalisations. By employing this security definition and using one-way communication only, we prove secrecy of the key under the most general eavesdropping attack allowed by the laws of quantum mechanics, so-called {\it coherent attacks}~\cite{CIRAC19971, PhysRevA.59.4238}, independent of the context in which the key is used.\\
In the asymptotic limit, i.e. for infinitely
many rounds,
%$n\rightarrow \infty$
the secret fraction $r_{\infty}$, i.e. the ratio of secret bits and the number of shared states (without parameter estimation rounds), is given by
\begin{equation}
\label{eq:secretfraction}
r_{\infty} = \sup_{U \leftarrow  K}\inf_{\sigma_{A\{B_i\}} \in \Gamma}
[S (U|E) - \max_{i\in\{1,...N-1\}} H( U|K_i)],
\end{equation}
where $U \leftarrow K$ denotes a bitwise preprocessing channel on Alice's raw key bit $K$,
$S(U|E)$ is the conditional von-Neumann entropy of the (classical) key variable, given the state of Eve's system $E$,
$H(U|K_i)$ is the conditional Shannon entropy of $U$ given $K_i$, which is $B_i$'s guess of $K$, and $\Gamma$ is the set of all density matrices $\sigma_{A\{B_i\}}$ of Alice and the Bobs which are consistent with the parameter estimation. The secret key rate is
\begin{equation}
R=r_\infty R_{\mathrm{rep}},\label{eq:keyrate}
\end{equation}
where the repetition rate $R_{\mathrm{rep}}=\frac{1}{t_{\mathrm{rep}}}$ is given by the time $t_{\mathrm{rep}}$ that one round (steps 1 and 2) takes. For now we set $t_{\mathrm{rep}}=1$. The secret key rate in Eq.~(\ref{eq:keyrate}) as a figure of merit does not directly account for the amount of needed local randomness, classical communication, qubits and gates. Depending on the context one might want to incorporate one or more of the former quantities into a cost-performance ratio as a figure of merit~\cite{PhysRevLett.112.250501}.\\
 Note that we have not assumed any symmetry about the quality of the channels connecting $A$ and $B_i$.
Therefore,
the worst-case
information leakage in the error correction step is determined by the noisiest channel, see the maximisation in the last term
of Eq. (\ref{eq:secretfraction}).
This is the main difference with respect to the bipartite case.\\

\subsection{The secret key rate} 
We now derive an analytical formula for the multipartite secret key rate based on a variant of the method of depolarisation~\cite{RGK05}. In practice, the described depolarisation operations will be applied to the classical data only, as described in detail below. Readers who are not interested in the technical details can skip to Eq.~(\ref{eq:rdep}).\\
Let us denote the GHZ basis of $N$ qubits as follows:
 \begin{equation}
 \label{ghz-basis}
 \ket{\psi_j^\pm} = \frac{1}{\sqrt{2}}(\ket{0}\ket{j}\pm \ket{1}\ket{\bar j})\ ,
 \end{equation}
 where $j$ takes the values $0,...,2^{N-1}-1$ in binary notation,
 and  $\bar j$ denotes the binary negation of $j$; i.e. for example
 if  $j = 01101$
 then ${\bar j} = 10010$. \\
 Remember that any state of $N$ qubits can be depolarised to a state which is diagonal in the GHZ basis by a sequence of local operations~\cite{DCT99,Duer+00}.
 In our protocol we introduce the following \textit{extended depolarisation procedure}. The set of depolarisation operators is
 \begin{equation}
 \mathcal{D}=\{X^{\otimes N}\} \cup \{Z_A Z_{B_j}|1\leq j \leq N-1 \} \cup \{R_k |1\leq k \leq N-1 \}, \label{eq:depolarisationoperators}
 \end{equation}
 where $X$ and $Z$ are Pauli operators and
 \begin{equation}
 R_k=\mathrm{diag}(1,i)_A \otimes \mathrm{diag}(1,-i)_{B_k}.
 \end{equation}
 The parties apply each of these operators with probability $1/2$ or $\1$ else.
 The operators from the first two sets of
 Eq.~(\ref{eq:depolarisationoperators}) make the density matrix GHZ diagonal as in~\cite{DCT99,Duer+00}. We denote the coefficient in front of $\proj{\psi_j^\sigma}$ by $\lambda_j^\sigma$ with $\sigma\in\{+,-\}$ and $j\in\{0,1,...,2^{N-1}-1\}$. The effect of $R_k$ is
 \begin{alignat}{2}
R_k \ket{\psi_j^\sigma}=&\left\{  \begin{array}{ll}
 \ket{\psi_j^\sigma} & \text{if } j^{(k)} = 0\\
 -i \ket{\psi_j^{-\sigma}} & \text{if } j^{(k)} = 1\\
 \end{array}\right. ,%\\
% &\Rightarrow \quad& R_k \ket{\psi_j^\sigma}\bra{\psi_j^\sigma} R_k^\dagger=&\left\{  \begin{array}{ll}
 %\ket{\psi_j^\sigma}\bra{\psi_j^\sigma} & \text{if } j^{(k)} = 0\\
 %\ket{\psi_j^{-\sigma}}\bra{\psi_j^{-\sigma}} & \text{if } j^{(k)} = 1\\
 %\end{array}\right.\\
 \intertext{so applying this operator with probability $\frac{1}{2}$ transforms}
% &\Rightarrow \quad& \lambda_j^\sigma \xrightarrow{R_k}&\left\{  \begin{array}{ll}
% \lambda_j^\sigma & \text{if } j^{(k)} = 0\\
% \lambda_j^{-\sigma} & \text{if } j^{(k)} = 1\\
% \end{array}\right.\\
  \lambda_j^\sigma \xrightarrow{}&\left\{  \begin{array}{ll}
 \lambda_j^\sigma & \text{if } j^{(k)} = 0\\
 \frac{1}{2}(\lambda_j^{-\sigma}+\lambda_j^{\sigma}) & \text{if } j^{(k)} = 1\\
 \end{array}\right. ,
 \end{alignat}
 where $j^{(k)}$ denotes the $k$th bit of the string $j$. As this operation is applied for all $k=1,2,...,N-1$, it achieves that
 \begin{equation}
 \lambda_j^+=\lambda_j^- \text{ for all $j>0$.}\label{eq:equallambdas}
 \end{equation}
 The resulting depolarised state reads
 \begin{equation}
 \label{ghz-diagonal}
 \rho_\text{dep} = \lambda_0^+ \proj{\psi_0^+} + \lambda_0^- \proj{\psi_0^-}+\sum_{j=1}^{2^{N-1}-1}\lambda_{j}(\proj{\psi_j^+}+\proj{\psi_j^-}).
 \end{equation}
 In our multipartite scenario we define the qubit error rate (QBER) $Q_Z$ to be the probability that at least one Bob obtains a different outcome than Alice in a $Z$-basis measurement. Note that this value is not the same as the bipartite qubit error rate $Q_{AB_i}$, which is the probability that the $Z$-measurement outcome of $B_i$ disagrees with the one of Alice. $Q_Z$ can be read directly from the structure of the depolarised state in Eq.~(\ref{ghz-diagonal}) and is given by
 \begin{equation}
 Q_Z = 1- \lambda_0^+-\lambda_0^- \ .\label{eq:QBER}
 \end{equation}
 For simplicity we neglect the possibility of increasing the key rate by adding pre-processing noise, i.e. we set $q=0$ in the notation of \cite{RGK05} such that $\mathbf{U}=\mathbf{K}$.
 Because
 \begin{equation}
 \begin{aligned}
 S(K|E)=&S(E|K)-S(E)+H(K)\\
 \text{and }H(K|K_i)=&H(K_i|K)-H(K_i)+H(K)
 \end{aligned}
 \end{equation}
 the asymptotic secret fraction is
 \begin{equation}
 \label{eq:askeyrate}
 r_{\infty} = S(E|K) - S(E) - \max_{1\leq i\leq N-1} (H(K_i|K) - H(K_i)).
 \end{equation}
 Note that we did not need to include the infimum over $\Gamma$, see Eq.~(\ref{eq:secretfraction}), here because, as we will see below, the measurement statistics completely determine all relevant quantities in our protocol.
 The entropies involving the classical random variable $K$ are directly obtained from the measurement statistics in the parameter estimation phase. They are given by
 \begin{equation}
 H(K|K_i)=h(Q_{AB_i}),
 \end{equation}
 with the binary Shannon entropy
 \begin{equation}
 h(p)=-p \log_2 p-(1-p) \log_2(1-p)
 \end{equation}
 and the bipartite error rate $Q_{AB_i}$, given by
 \begin{equation}
 Q_{AB_i}=\sum_{\substack{j\\ j^{(i)}=1}}\sum_{\sigma=\pm} \lambda_j^\sigma
 \overset{Eq.~(\ref{eq:equallambdas})}{=} 2\sum_{\substack{j\\ j^{(i)}=1}}\lambda_j,
 \end{equation}
 where $j^{(i)}$ denotes the $i$-th bit of $j$ and, because both outcomes are equiprobable, 
 \begin{equation}
 H(K_i)=1.
 \end{equation}
 Giving Eve the purification of Eq.~(\ref{ghz-diagonal}), the von-Neumann entropies involving Eve's system in Eq.~(\ref{eq:askeyrate}) are given by
 \begin{align}
 S(E|K)\overset{\phantom{Eq.~(\ref{eq:equallambdas})}}{=}&\frac{1}{2} S(E|K=0)+ \frac{1}{2} S(E|K=1)\nonumber \\
 %\overset{\phantom{Eq.~(\ref{eq:equallambdas})}}{=}&\frac{1}{2} S(\sigma_E^0)+\frac{1}{2} S( \sigma_E^1 )\nonumber\\
 \overset{\phantom{Eq.~(\ref{eq:equallambdas})}}{=}&-\sum_{i=0}^{2^{N-1}-1} (\lambda_i^+ + \lambda_i^-)\log_2 (\lambda_i^+ + \lambda_i^-)\nonumber\\
 \overset{Eq.~(\ref{eq:equallambdas})}{=}%&-(1-Q_Z)\log_2(1-Q_Z)-2\sum_{i=1}^{2^{N-1}-1} \lambda_i(1+\log_2 (\lambda_i))\nonumber\\
 %\overset{\phantom{Eq.~(\ref{eq:equallambdas})}}{=}
 &-(1-Q_Z)\log_2(1-Q_Z)-2\sum_{i=1}^{2^{N-1}-1} \lambda_i\log_2 (\lambda_i) - Q_Z
 \label{eq:SEUnosymmetry}
 \end{align}
 and
 \begin{align}
 S(E)\overset{\phantom{Eq.~(\ref{eq:equallambdas})}}{=}&S(\frac{1}{2}(\sigma_E^0 + \sigma_E^1))%\nonumber\\
 %\overset{\phantom{Eq.~(\ref{eq:equallambdas})}}{=}&
 =-\sum_{j,\sigma=\pm} \lambda_j^{\sigma} \log_2 \lambda_j^{\sigma}\nonumber\\
 \overset{Eq.~(\ref{eq:equallambdas})}{=}&-\lambda_0^+\log_2\lambda_0^+-\lambda_0^-\log_2\lambda_0^- -2 \sum_{j>0} \lambda_j \log_2 \lambda_j
 \label{eq:SEnosymmetry},
 \end{align}
 i.e.
 \begin{align}
 S(E|K)-S(E)%=&-Q_Z-(1-Q_Z)\log_2(1-Q_Z)+\lambda_0^+\log_2\lambda_0^++\lambda_0^-\log_2\lambda_0^- \\
 =&-Q_Z-(1-Q_Z)\log_2(1-Q_Z)+\lambda_0^+\log_2\lambda_0^++(1-Q_Z-\lambda_0^+)\log_2(1-Q_Z-\lambda_0^+). \label{eq:difference}
 \end{align}
 Now $\lambda_0^+$ and $\lambda_0^-$ can be obtained with the additional $X^{\otimes N}$ measurement in the parameter estimation, because $\lambda_0^++\lambda_0^-=1-Q_Z=\tr \left(\rho_{\mathrm{dep}} (\proj{0}^{\otimes N}+\proj{1}^{\otimes N})\right)$ is known from the QBER and $\tr\left( \rho_{\mathrm{dep}} X^{\otimes N}\right) = \sum_j (\lambda_j^+-\lambda_j^-)=\lambda_0^+-\lambda_0^-$. In analogy to $Q_Z$ we denote the probability that the $X$-measurement gives an unexpected result, i.e. one that is incompatible with the noiseless state, by $Q_X$. Because $\bra{\psi_j^\sigma}X^{\otimes N}\ket{\psi_j^\sigma}=\sigma$ this leads to
 \begin{equation}
 Q_X=\frac{1-\langle X^{\otimes N}\rangle_{\mathrm{dep}}}{2},
 \end{equation}
 which can, as we will see in Section~\ref{sec:implementationofprotocol}, be obtained from the measured data in the parameter estimation step. We remark that $Q_X$ is not the probability that at least one Bob gets a
 different $X$-measurement outcome than Alice, as the outcomes are
 not correlated, see Appendix~\ref{app:GHZ}.\\
 Finally, inserting Eq.~(\ref{eq:difference}) into Eq.~(\ref{eq:askeyrate}), and using Eq.~(\ref{eq:keyrate}),
 we arrive at the achievable secret key rate,
\begin{equation}
\begin{aligned}
R=&\hphantom{{}+{}}\left(1- \frac{Q_Z}{2} - Q_X\right) \log_2\left(1- \frac{Q_Z}{2} - Q_X\right)
 + \left(Q_X-\frac{Q_Z}{2}\right) \log_2\left(Q_X-\frac{Q_Z}{2}\right) \\& + (1-Q_Z) (1 - \log_2(1-Q_Z))- h(\max_{1\leq i \leq N-1}Q_{AB_i}).
\end{aligned}\label{eq:rdep}
\end{equation}
Note that the parameters in this equation are obtained from the measured data and will depend on the number of parties $N$.\\

\section{Implementation and noise}\label{sec:implementation}
In this section we compare the multipartite-entanglement-based protocol (NQKD) as introduced above to a protocol based on bipartite entanglement (2QKD), which we define in the following.
\subsection{Conference key distribution with bipartite entangled quantum states (2QKD)}
A suitable protocol
to establish
a secret joint key for $N>2$  parties via bipartite entanglement
proceeds as follows, see Fig.~\ref{fig:bipartite}:
Party $A$ shares a Bell
state with each of the $N-1$ parties $B_i$ and establishes a
(different)
secret bipartite key ${\bf S}_i$ with each party $B_i$. For concreteness, we assume in our comparison that the six-state protocol~\cite{Bruss98} is used.
In general, the
$N-1$ channels may be different and thus have individual
QBERs. Party $A$ then defines a new random key ${\bf k}_c$ to be
the conference key. She sends the encoded conference key
${\bf k}_i = {\bf S}_i \oplus {\bf k}_c$ to party $B_i$ who
performs ${\bf k}_i \oplus {\bf S}_i = {\bf k}_c$ and thus regains
the conference key.\\
A comparison of the performance of the bipartite versus the multipartite
entanglement-based strategy for
$N$ parties is subtle and has to consider various aspects,
as different resources are needed: on one hand only  bipartite
entanglement is needed for 2QKD, while multipartite
entanglement is needed for NQKD.
(Note, however, that the number of necessary two-qubit
gates for generation of the entangled states is in both cases $N-1$.)
On the other hand, the number of resource
qubits per round is $2(N-1)$ for 2QKD, while
only $N$ qubits are needed for NQKD.
Finally, the 2QKD protocol requires to transmit $(N-1)$ additional classical bits (the encoded conference key). Thus, each of the two strategies
has its own advantages. A quantitative comparison regarding imperfections in preparation and transmission is discussed below.\\

\subsection{Implementation of the NQKD protocol} \label{sec:implementationofprotocol}
We now describe how the depolarisation operations used in the security proof can effectively be implemented classically by adjusting the protocol.\\
For key generation and the $Q_Z$ estimation, the parties perform $Z^{\otimes N}$-measurements. These are only affected by the $X^{\otimes N}$ depolarisation operator, which flips the outcomes of all parties. It can therefore be implemented on the classical data. The other depolarisation operators are diagonal in the $Z$-basis and thus do not change the $Z$-measurement outcome.\\
Let us call the parameter estimation rounds, in which the parties measure $X^{\otimes N}$ (after depolarisation), estimation rounds of the second type. How the depolarisation step affects the $X^{\otimes N}$-measurement is not so obvious and is described in the following.\\
Note that the depolarisation operators $X^{\otimes N}$ and $Z_AZ_{B_k}$, $k=1,2,...,N-1$ (see Eq.~(\ref{eq:depolarisationoperators})), commute with the $X^{\otimes N}$-measurement and thus these depolarisation operators do not have an effect in second type rounds. But
\begin{equation}
R_k X_A X_{B_k} R_k^\dagger = (-Y_A) Y_{B_k}
\end{equation}
i.e. applying the depolarisation operator $R_k$ is equivalent to Bob $k$ measuring in $Y$-basis. Also note that
\begin{equation}
R_k (-Y_A) X_{B_k} R_k^\dagger = (-X_A) Y_{B_k},
\end{equation}
so let $\kappa_j$ be the number of Bobs measuring in $Y$-basis in the $j$-th round, then Alice measures in the basis
\begin{equation}
M_A(\kappa) = \left\{\begin{array}{cl}
X_A & \text{if } \kappa_j \bmod 4 = 0\\
-Y_A & \text{if }\kappa_j \bmod 4 = 1\\
-X_A & \text{if }\kappa_j \bmod 4 = 2\\
Y_A & \text{if }\kappa_j \bmod 4 = 3
\end{array}
\right. ,
\end{equation}
where a minus sign corresponds to a flip of the measurement outcome. Note that this measurement rule for Alice implies that always an even number of parties measures in $Y$-basis and that the outcome of the measurement is flipped whenever it is not a multiple of four. 
Each party measures in $X$ or $Y$ basis with probability $1/2$. Note that the rule for $M_A$ described above means that only half of all possible combinations of these measurement bases are actually measured. However, in practice the parties can measure $X$ and $Y$ independently with probability $1/2$ and throw away half of their data (where an odd number of parties has measured in $Y$-basis) and Alice still flips her measurement outcome whenever the number of parties measuring in $Y$-basis was not a multiple of four. This is not a problem, because in the parameter estimation rounds each party announces its measurement setting and outcome. We thus arrive at the implementation described initially.
Let $\tilde{\kappa}_j$ be the number of parties measuring in $Y$-basis in run $j$, i.e.
\begin{equation}
\tilde{\kappa}_j = \left\{\begin{array}{cl}
\kappa_j + 1 & \text{if Alice measured in $Y$-basis}\\
\kappa_j & \text{else}
\end{array} \right.,
\end{equation}
then
\begin{align}
\langle X^{\otimes N} \rangle_{\mathrm{dep}} =&\lim_{\text{\#exp}\rightarrow\infty} \frac{1}{\text{\#exp}} \sum_{j=1}^{\text{\#exp}} f(\tilde{\kappa}_j) \prod_{i=1}^N a_{i,j} \\
=&\lim_{\text{\#exp}\rightarrow\infty} \frac{n_+-n_-}{n_++n_-} , \label{Eq:DetOfXN}
\end{align}
where \#exp is the number of experiments in the second type rounds with even $\tilde{\kappa}_j$, $a_{i,j}$ is the outcome of party $i$ in experiment $j$, 
\begin{equation}
f(\tilde{\kappa})=\left\{\begin{array}{cl}
0 & \text{if $\tilde{\kappa}_j$ odd}\\
1 & \text{if }\tilde{\kappa}_j \bmod 4 = 0\\
-1 & \text{else}
\end{array}\right. \label{eq:signofkappa}
\end{equation}
and
\begin{equation}
n_{\pm} = \frac{1}{2} \text{\#exp}\pm \frac{1}{2}\sum_{j=1}^{\text{\#exp}} f(\tilde{\kappa})\prod_{i=1}^N a_{i,j}. 
\end{equation}
We remark that, in contrast to full tomography, the number of rounds needed to get sufficient statistics for estimating $\langle X^{\otimes N}\rangle_{\mathrm{dep}}$ does not increase with the number of parties $N$.\\[1ex]
Let us summarise the steps of an implementation of the NQKD protocol:
\begin{enumerate}
	\item Distribution of the state GHZ state $\ket{\psi_0^+}$.
	\item $L\cdot h(p_p)$ bits of pre-shared key are used to mark the second type rounds, where $L$ is the total number of rounds and $p_p$ is the probability for an $X^{\otimes N}$-round. This amount of key suffices, because an $L$-bit binary string with a $1$ for each second type round can asymptotically be compressed to $L\cdot h(p_p)$ bits.
	\item In each second type round each party measures randomly in the X- or $Y$-basis.
	\item In all other cases all parties measure in $Z$-direction.
	\item Parameter estimation:
	\begin{enumerate}
		\item Alice announces a randomly chosen small subset of size $L\cdot h(p_p)$ of $Z$-measurement rounds, in which all parties announce their $Z$-measurement results. From this data the QBER $Q_Z$ and the individual QBER's $Q_{AB_i}$ are estimated.
		\item The parties announce the measurement results of the second type rounds together with the chosen measurement basis. Alice flips her outcome if the number of parties who measured in $Y$-basis is not a multiple of four (see Eq.~(\ref{eq:signofkappa})). From the data where an even number of parties measured in $Y$-basis (including zero), the parameter $Q_X$ is calculated according to Eq.~(\ref{Eq:DetOfXN}).
	\end{enumerate}
	\item Alice announces which $Z$-measurement results all parties have to flip (the probability for each bit is $1/2$). This effectively implements the depolarisation with operator $X^{\otimes N}$.
	\item Classical post-processing:
	\begin{enumerate}
		\item Alice sends error correction information (for $\max_i Q_{AB_i}$) to all Bobs, which perform the error correction.
		\item In privacy amplification the parties obtain the key by applying a two-universal hash function, which was chosen randomly by Alice, to the error corrected data.
	\end{enumerate}
	\item The achievable key rate is then given by Eq.~(\ref{eq:rdep}).
\end{enumerate}

\subsection{Example of depolarising noise}
In this section we assume that $\rho_{AB_1...B_{N-1}}$ is a mixture of the GHZ-state and white noise, i.e. the parties share the state
\begin{equation}
\rho = \lambda_0^+ \proj{\psi_0^+} + \frac{1-\lambda_0^+}{2^N-1} (\1-\proj{\psi_0^+}). \label{eq:wernerstate}
\end{equation}
Here all coefficients other than $\lambda_0^+$ are equal, i.e. $\lambda_j^\pm = \lambda_0^- =Q_Z/(2^N-2)$ for $j=1,..., 2^{(N-1)}-1$ and $\lambda_0^+ = 1-Q_Z\frac{2^N-1}{2^N-2}$. The rate of unexpected results for the $X^{\otimes N}$-measurement is thus given by
\begin{equation}
Q_X=\frac{2^{N-2}}{2^{N-1}-1}Q_Z.
\end{equation}
For the highly symmetric state of Eq.~(\ref{eq:wernerstate}) the key rate is then a function of $Q_Z$ and $N$ only. The terms in Eq.~(\ref{eq:askeyrate}) are
\begin{align}
Q_{AB_i}=& \frac{2^{N-1}}{2^N-2} Q_Z,\\
S(E|U)=& -(1-Q_Z) \log_2(1-Q_Z) - Q_Z \log_2 \frac{2Q_Z}{2^N-2}\\
\text{and }S(E)=&-(1-Q_Z\frac{2^N-1}{2^N-2})\log_2 (1-Q_Z\frac{2^N-1}{2^N-2})\nonumber\\
& - (2^N-1)\frac{Q_Z}{2^N-2} \log_2 (\frac{Q_Z}{2^N-2})
\end{align}
and inserting them into Eq. (\ref{eq:askeyrate}) leads to the asymptotic secret key rate as
function of $Q=Q_Z$ and $N$,  namely
\begin{equation}%
\label{eq:rate}
R(Q,N)=1 + h(Q)- 
h\left(Q \frac{2^N - 1}{2^N - 2}\right) - h\left(Q \frac{2^{N - 1}}{2^N - 2}\right)
+    \left(\log_2(2^{N - 1} - 1) - \frac{2^N - 1}{2^N - 2} \log_2 (2^N - 1)\right) Q.
\end{equation}%
This function is shown in Fig.~\ref{fig:keyrates}.
\begin{figure}[tbp] %
	\centering %
	\subfigure[Key rates (Eq.~(\ref{eq:rate})) for $N=2,3,...,8$ (left to right) as a function of the QBER $Q_Z$ of a depolarised state (see Eq.~(\ref{eq:wernerstate})). The dashed line corresponds to the limit $N\rightarrow \infty$.\label{fig:keyrates}]{\includegraphics[width=0.49\linewidth]{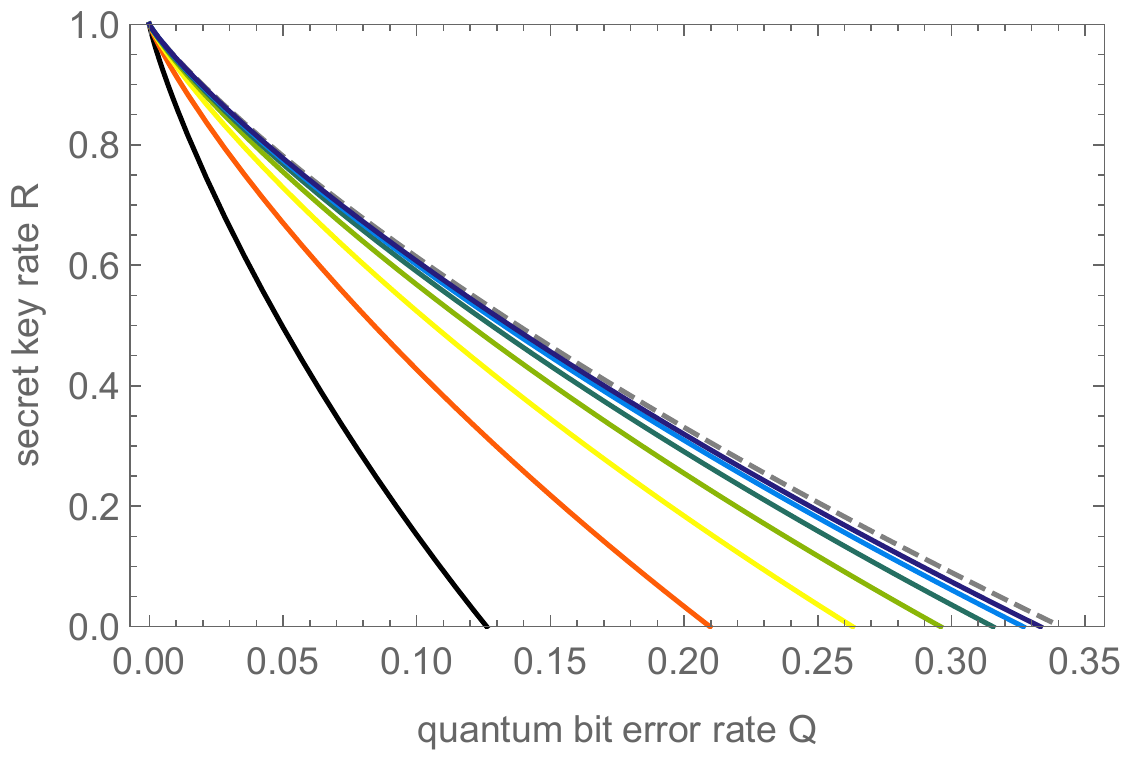}}\hfill%
	\subfigure[Key rates for $N=2,3,4,...,8$ parties (right to left) as a function of the two-qubit gate failure probability $f_G$.\label{fig:keyratesfG}]{\includegraphics[width=0.49\linewidth]{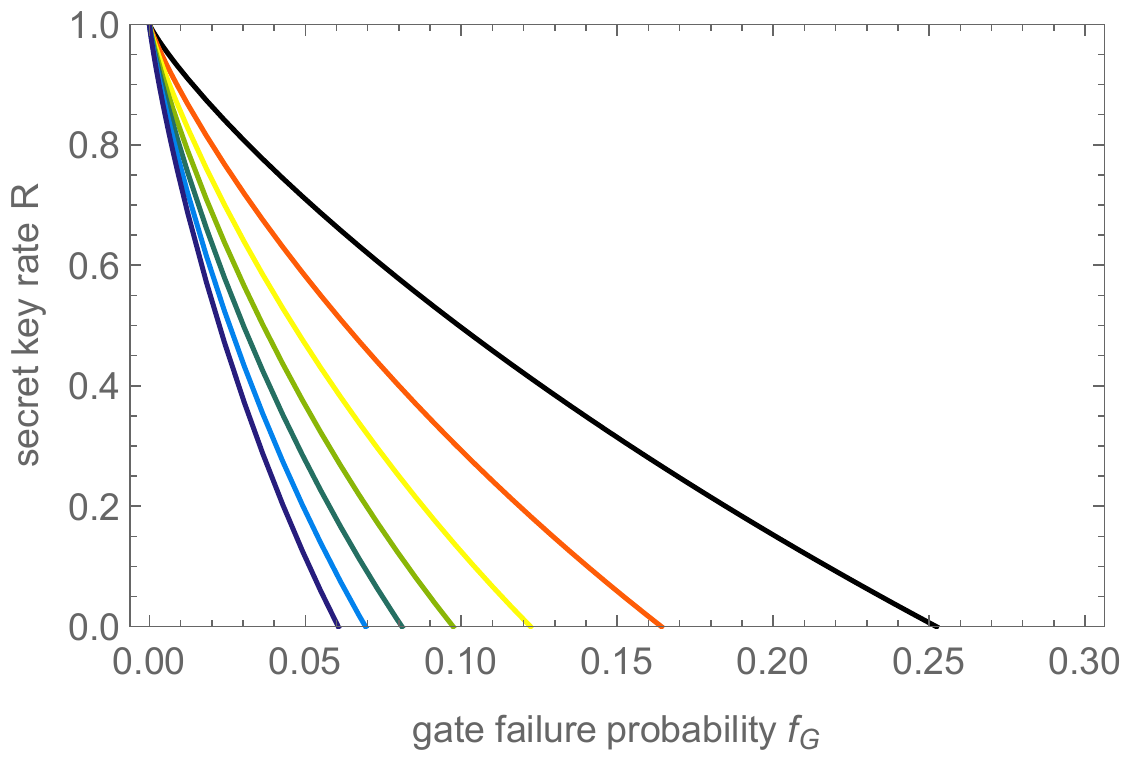}}
	\caption{The secret key rate of the NQKD protocol as a function of the QBER (a) and the gate failure probability (b).}\label{fig:keyratePlots}
\end{figure}%
For $N=2$ the key rate coincides with the one of the six-state protocol~\cite{Bruss98,RGK05}, namely
\begin{equation}
R(Q,2)= 1-h\left(\frac{3}{2}Q\right)-\frac{3 \log_2 3}{2} Q.
\end{equation}
In the limit of large $N$ the key rate simplifies to
\begin{equation}%
R(Q,\infty)=1-h\left(\frac{Q}{2}\right)-Q. %
\end{equation}%
We also numerically determined the threshold values for the QBER, i.e. the value of $Q$ until which a non-zero secret key rate is achievable, for different numbers of parties $N$, see Table~\ref{tab:threshold}.%
\begin{table}%
	\caption{Threshold values of the multipartite entanglement based protocol (NQKD) without preprocessing noise for different numbers of Parties $N$.
		The well-known bipartite case, i.e. $N=2$,  is also given for comparison.
		A non-zero secret key can be distilled if the QBER is below the listed value.}\label{tab:threshold}%
	\begin{tabular}{cc}%
		N & Threshold QBER\\%
		\hline%
		2 & 0.126193\\%
		3 & 0.209716\\%
		4 & 0.263087\\%
		5 & 0.295974\\%
		6 & 0.315562\\%
		7 & 0.326892\\%
		8 & 0.333296\\%
		9 & 0.336851\\%
		10 & 0.338799%
	\end{tabular}\hspace{1cm}\begin{tabular}{cc}%
		N & Threshold QBER\\%
		\hline%
		11 & 0.339855\\%
		12 & 0.340424\\%
		13 & 0.340728\\%
		14 & 0.340890\\%
		15 & 0.340976\\%
		16 & 0.341021\\%
		17 & 0.341045\\%
		\vdots  & \\%
		$\infty$ & 0.341071%
	\end{tabular}%
\end{table} %
Note that for fixed $Q$ the key rate increases with the number of parties $N$. However, one might expect that in practice the QBER is not constant but increases with increasing number of parties $N$ (because the experimental creation of the $N$-partite GHZ state becomes more demanding). This intuition is discussed quantitatively in the following section.\\

\subsection{Noisy gates and channels}
Let us compare the performance of NQKD and 2QKD when
using imperfect two-qubit gates in the production of
the entangled resource states.
We employ, for both 2QKD and NQKD, the model of depolarising noise, i.e. if a two-qubit gate fails,
which happens with probability $f_G$, then the two processed qubits
are traced out and replaced by the completely mixed state.\\
 When the GHZ resource state is produced
in the network of Fig.~\ref{fig:multipartite},
Alice starts with the state $\ket{+}_A\ket{0}^{\otimes N-1}$ and applies a controlled-NOT gate from $A$ to each of the other qubits.

The secret key rate is shown in Fig.~\ref{fig:keyratesfG} as a function of the gate error rate $f_G$. It captures the expectation that the demands on the gates for producing an $N$-party
GHZ state increase with the number of parties $N$.
We mention that the GHZ state could also be produced using a single multi-qubit gate, e.g. $C_{X^{\otimes (N-1)}}=\proj{0}\otimes\1+\proj{1}\otimes X^{\otimes (N-1)}$, which is locally equivalent to the controlled-Phase gate, see e.g.~\cite{Liu2016}. The QBER caused by this gate is $Q=\frac{f_G}{2}$. Because the threshold $Q$ increases with $N$ (cf. Fig.~\ref{fig:keyrates}), so does the threshold gate failure probability in this case.

In addition to imperfect gates, noise might be introduced by the transmission channel. Consider, for example, the situation when the qubit of each Bob is individually affected by a depolarising channel. Let the probability of depolarisation be $f_C$, then the QBER is
\begin{equation}
Q(f_C)=\frac{2^N-2 }{2^N}\left(1-(1-f_C )^N\right) \label{eq:fC}
\end{equation}
and the key rate can be calculated according to Eq.~(\ref{eq:rate}).\\

\section{Quantum key distribution in networks}\label{sec:networks}
We will now show that in quantum networks with constrained channel
capacity and with quantum routers, employing
multipartite entanglement leads to
a higher secret key rate than bipartite entanglement, when
the gate quality is higher than a threshold value.\\
Beyond the simple network of Fig.~\ref{fig:schemes}, the GHZ resource state can be distributed in many different networks. Consider a fixed but general network as given via a graph with vertices and directed edges. Let all channels have the same
transmission capacity (also called bandwidth), which is associated with the direction of the corresponding edge. For the sake of a simple presentation, we assume that this transmission capacity is one qubit per second. Thus, the time $t_{\mathrm{rep}}$ consumed in one round (steps 1 and 2 of the protocol) is proportional to the number of network uses in that round. %For simplicity we now set the time unit to one network use, i.e. from now on $t_{\mathrm{rep}}$ is the number of network uses per round.
A generic network has some bottlenecks. In this case the difference between the NQKD and 2QKD  protocol becomes evident: Alice may send a single qubit in the NQKD scheme, while she has to transmit $N-1$ qubits in the 2QKD case.\\
As an example consider the quantum network where all parties are connected to a single central router $C$, see Fig.~\ref{fig:router}. %
\begin{figure}[tbp]%
	\includegraphics[scale=1.0]{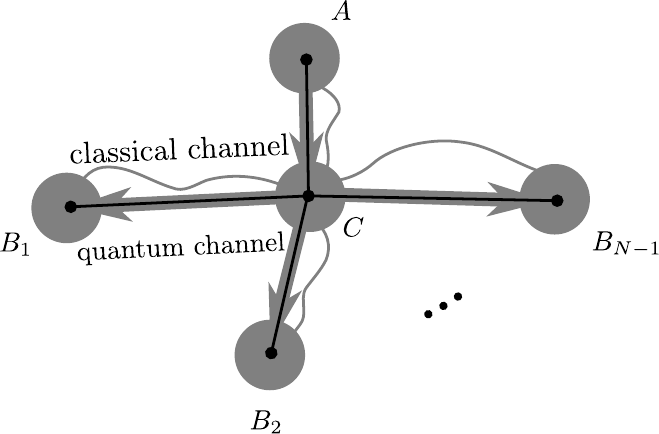}%
	\caption{This quantum network with a central router $C$, which is able to
produce and entangle qubits, exemplifies a network with a bottleneck. The GHZ-like resource state used in the multipartite entanglement QKD protocol, see Eq.~(\ref{eq:GHZ}), can be distributed in a single use of the depicted network (i.e. each channel transmits a single qubit only)~\cite{EKB16b}, while $N-1$ uses of the network are necessary in the 2QKD protocol.}\label{fig:router}%
\end{figure}%
Because $C$ is not trusted we assume it to be in the control of Eve. In this network
the channel from $A$ to $C$ constitutes a bottleneck. Note, however, that this network can be much more economical than the one of Fig.~\ref{fig:schemes} if the distance between $A$ and $C$ is large. The 2QKD protocol needs $N-1$ network uses, i.e. $t_{\mathrm{rep}}^{\mathrm{(2QKD)}}=(N-1) \,\mathrm{s}$,  to distribute the Bell pairs. In contrast to this the NQKD protocol can employ the quantum network coding~\cite{Ahlswede00,Leung06,Hayashi07,Kobayashi09,Kobayashi10,Beaudrap14,Satoh16} scheme of reference~\cite{EKB16b} to distribute the GHZ state in a single network use, i.e. $t_{\mathrm{rep}}^{\mathrm{(NQKD)}}=1\,\mathrm{s}$. See Appendix~\ref{app:QNC} for the explicit calculation. Thus the key rate of the NQKD protocol is $(N-1)$ times larger than the one of the 2QKD protocol in the ideal case ($r_\infty=1$).\\
When again using noisy two-qubit gates (the QBER calculation is analogous to the case of the network shown in Fig~\ref{fig:multipartite} discussed above), the QBER for the NQKD protocol increases with $N$. These two effects lead to gate error thresholds below which the NQKD protocol outperforms 2QKD, see Fig.~\ref{fig:maxfGNQKDadvantage}.
\begin{figure}[tbp]%
	\centering%
	\subfigure[Preparation noise, see Appendix~\ref{app:gates} for details.\label{fig:maxfGNQKDadvantage}]{\includegraphics[width=0.49\linewidth]{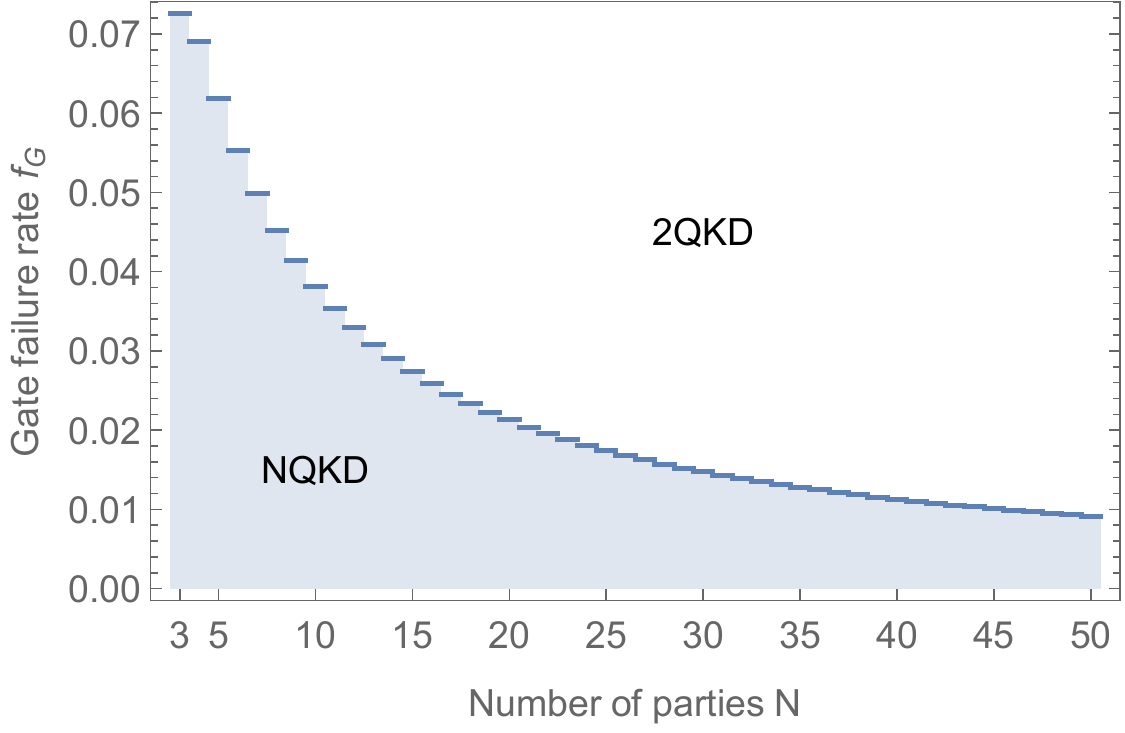}}\hfill%
		\subfigure[Transmission noise, see Eq.~(\ref{eq:fC}).\label{fig:maxfCNQKDadvantage}]{\includegraphics[width=0.49\linewidth]{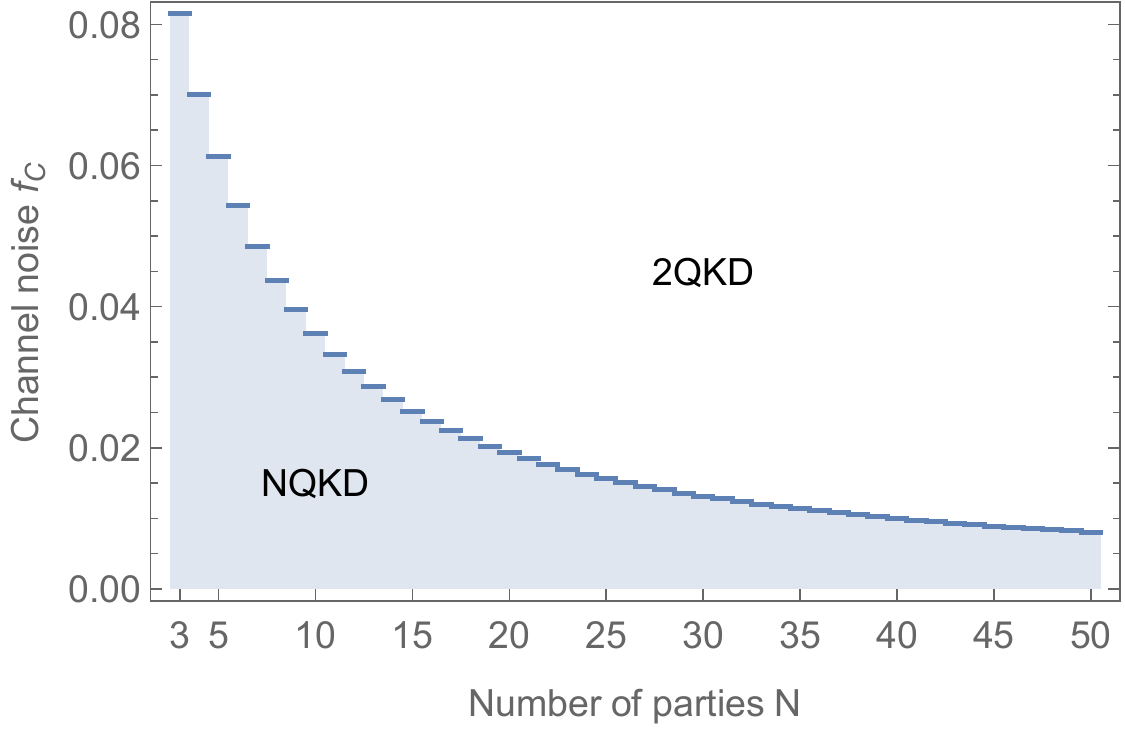}}%
	\caption{For less noise than the shown threshold, i.e. in the blue area, NQKD leads to higher key rates than 2QKD in the network of Fig.~\ref{fig:router}.} \label{fig:NQKDadvantage}%
\end{figure}%
For a fixed number of parties $N$ there is a maximal gate error probability below which the NQKD protocol outperforms the bipartite approach in the quantum network of Fig.~\ref{fig:router}. For $N=3$ already gate failure rates below $7.2\,\%$ imply that NQKD outperforms 2QKD. More values are listed in the Appendix~\ref{app:gates}.\\
The exact same behavior can be observed when considering noisy channels. In the ideal case NQKD outperforms 2QKD, while NQKD is more prone to channel noise. The resulting threshold noise levels are shown in Fig.~\ref{fig:maxfCNQKDadvantage}.\\
We mention that the famous butterfly network~\cite{Ahlswede00} leads to a similar advantage, see Appendix~\ref{app:butterfly} for details.
%\FloatBarrier
\section{Conclusion}
In this paper we analysed
a quantum conference key distribution (QKD)  protocol for $N$ parties
which is based on multipartite entangled resource states.
We generalised the information theoretic security analysis of \cite{Renner05} to this $N$-partite scenario.
Using the depolarisation method we derived an
 analytical formula for the secret key rate as a function of the quantum bit error rate (QBER). For a fixed QBER the secret key rate is found to
increase with the number of parties. Accordingly, the threshold QBER until which a non-zero secret key can be obtained increases with the number of parties.
\\
Furthermore, we presented an example where multipartite entanglement-based
QKD outperforms the approach based on bipartite QKD links.
We found this advantage in networks with bottlenecks and showed that it
holds above a certain threshold gate quality which depends on the number
of parties.\\
We expect more interesting insights from analysing further aspects
 of the multipartite entanglement-based
QKD protocol. Regarding implementations the secret key calculation of the protocol for finite numbers of rounds will be beneficial. Various examples
of network layouts and the link to network coding schemes will deserve more detailed investigations.

\bibliographystyle{apsrev4-1}
\bibliography{references}
\vspace*{3ex}

\noindent {\bf\large Acknowledgments}\\
We acknowledge helpful discussions with Jan B\"orker and Norbert L\"utkenhaus. This work was financially supported by BMBF (network Q.com-Q) and ARL.\\

\appendix
\section{The resource state and its properties}\label{app:GHZ}
In this section we
derive the form of a pure quantum state that fulfils the requirements of
perfect correlations for one set of local measurement bases, with
uniformly distributed random measurement outcomes.
(These local bases are used for the key generation.)
%Without loss of generality the local basis of each party is
%denoted as $Z$-basis.
We also prove  properties of the resource state
regarding correlations of measurement outcomes in any other set
of local bases.\\
A general normalized $N$-qubit state reads
\begin{equation}
\ket{\phi} = \sum_{i_1,i_2,...i_N=0}^1 a_{i_1,i_2,...i_N}\ket{i_1,i_2,...i_N}\ ,
\end{equation}
with complex coefficients $a_{i_1,i_2,...i_N}$
that satisfy
$\sum_{i_1,i_2,...i_N=0}^1 |a_{i_1,i_2,...i_N}|^2=1$.
To achieve perfect correlations,
we can assume without loss of generality that all parties measure
in the $Z$-basis and  get
the same outcome, as the choice of another local basis corresponds
to a local rotation,
and an opposite outcome could
be flipped locally. The requirement of perfect correlations in the $Z$-basis is only fulfilled by a
quantum correlated state of the form
\begin{equation}
\ket{\phi_{corr}} = a_{0,...,0}\ket{0,...,0}+
a_{1,...,1}\ket{1,...,1}\ .
\label{resource}
\end{equation}
%, and we therefore conjecture them to perfom worse.

It turns out that this requirement of perfect correlations in
one set of local bases forbids perfect correlations,
even only pairwise, in any other local bases,
for all $N\ge 3$.
\begin{theorem}
	For $N$ qubits, with
	$N\ge 3$, the state $\ket{\phi_{corr}} = a_{0,...,0}\ket{0,...,0}+
	a_{1,...,1}\ket{1,...,1}$  leads to perfect classical correlations
	between any number of parties, if and only if each of them
	measures in the $Z$-basis.
\end{theorem}
\begin{proof}
	Measuring in the $Z$-basis, perfect correlations follow
	trivially. For the reverse implication,
	let us denote the direction of measurement  for party $i$ by the vector
	$\vec{M_i}$, with components $M_{i}^{x}, M_{i}^{y}$ and $M_{i}^{z}$.
	An observable ${\cal M}_{ij}$ of two parties $i$ and $j$ is given by
	\begin{equation}
	\label{twopartymeas}
	{\cal M}_{ij} = (\vec{M_i}\cdot \vec{\sigma}) \otimes (\vec{M_j}\cdot \vec{\sigma})
	= \sum_{\alpha,\beta\in \{x,y,z\}}M_{i}^{\alpha}M_{j}^{\beta}\sigma_i^\alpha \otimes
	\sigma_j^\beta ,
	\end{equation}
	where $\vec{\sigma}$ denotes the vector of Pauli matrices and the identity
	operators for the parties $\neq i,j$ are omitted.
	Observe that
	\begin{equation}
	\label{expectvalue}
	\bra{\phi_{corr}}\sigma_i^\alpha \otimes
	\sigma_j^\beta\ket{\phi_{corr}} = 0 \ \ \text{unless} \ \ \alpha=\beta=z ,
	\end{equation}
	%where the identities at all positions $\neq i,j$ are not explicitly written,
	because all other combinations of Pauli operators change $\ket{\phi_{corr}}$ to
	an orthogonal state.
	
	Denoting by $p_i^{\alpha}(\pm)$
	the probability that party $i$ finds eigenvalue
	$\pm 1$ when measuring $\sigma^\alpha$,
	we also have $\bra{\phi_{corr}}\sigma_i^\alpha \otimes
	\sigma_j^\beta\ket{\phi_{corr}}= 2[p_i^{\alpha}(+)p_j^{\beta}(+)+
	p_i^{\alpha}(-)p_j^{\beta}(-)] -1$, and thus
	$p_i^{\alpha}(+)p_j^{\beta}(+)+
	p_i^{\alpha}(-)p_j^{\beta}(-) \neq 1$, unless
	$\alpha = \beta = z$.
	Therefore, perfect correlations between two parties
	are not possible in any other than the $Z$-basis.
	This also excludes perfect correlations
	between any other number of parties.
	- Note that the above argument, in particular Eq.~(\ref{expectvalue}), does not hold for
	$N=2$, which is special.
\end{proof}

Thus, any state of the form (\ref{resource}) contains the resource of perfect
multipartite correlations in the local $Z$-bases.
In order to ensure uniformity
of the outcome, i.e. randomness of the resulting secure bit string,
we choose  for the key generation protocol $|a_{0,...,0}| = 1/\sqrt{2} =
|a_{1,...,1}|$, i.e. the unique perfect resource is a GHZ state~\cite{GHZ07}.

\section{Security analysis of the NQKD protocol}\label{app:security}
In this appendix we generalise the composable security definition of the bipartite
scenario~\cite{Renner05,Sca+09} to  the $N$-partite case. As mentioned in the main text, the security analysis proceeds along analogous lines as the bipartite case in \cite{RGK05,RGK05PRL}.  We assume that the parties $A$ and $B_i$, for $i=1,...,N-1$  share $n$ multipartite
states. The eavesdropper $E$  is supposed to hold a purification of the global state. The total quantum state after $Z$-measurement of $A$ and all $B_i$ 
is described by the density operator
\begin{equation}
\begin{aligned}
\rho^n_{{\bf K K_1...K_{N-1}}E} =& \sum_{\bf{x,x_1,...,x_{N-1}}}
P_{{\bf K K_1...K_{N-1}}}(\bf{x,x_1,...,x_{N-1}}) \\
&\proj{\bf{x}}\otimes
\bigotimes_{i=1}^{N-1}\proj{\bf{x_i}}\otimes
\rho_E^{\bf x,x_1,...x_{N-1}} \ ,
\end{aligned}
\end{equation}
where ${\bf x}$ and  ${\bf x_i}$
describe the classical strings of parties A and $B_i$, respectively.
Note that the classical post-processing is identical to the bipartite
case: In an error correction step the parties transform their only partially correlated 
raw data into a fully correlated shorter string. Party A pre-processes her random string ${\bf K}$ 
according to the channel
${\bf U}\leftarrow {\bf K}$ and sends classical error correction
information ${\bf W}$ to parties
$B_i$, who compute their respective guesses ${\bf U_i}$ for ${\bf U}$
from ${\bf K_i}$ and ${\bf W}$. The error correction information ${\bf W}$ is the same for all Bobs, 
thus there is no additional information leakage compared to the bipartite case.
In a second step, the privacy amplification, Party A randomly chooses $f$ from a two-universal family 
of hash functions, computes her key  ${\bf S_A}= f({\bf U})$
and sends the description of $f$ to all
parties $B_i$ who also perform the privacy amplification to arrive at
their respective keys  ${\bf S_{B_i}}=
f({\bf U_i})$. The total quantum
state will then be denoted as $\rho_{{\bf S_A S_{B_1}...S_{B_{N-1}}}E`}$. The key tuple
(${\bf S_A,  S_{B_1},..., S_{B_{N-1}}}$) is called
$\epsilon$-secure, if it is $\epsilon$-close to the ideal state, i.e.\
if
\begin{equation}
\delta(\rho_{{\bf S_A S_{B_1}...S_{B_{N-1}}}E`}, \rho_{\bf SS...S}
\otimes \rho_{E`}) \leq \epsilon \ ,
\end{equation}
where $\delta(\rho,\sigma)=  \text{tr}|\rho - \sigma|/2$
denotes the trace distance.

Note that we have not assumed any symmetry about the quality
of the channels connecting A and $B_i$. The information leaking
to the eavesdropper in the error correction step is determined by the amount of error correction information which the 
Bob with the noisiest channel requires. This is the main difference with respect to the bipartite
case.

Therefore we arrive at the following
key length $\ell^{(n)}$, generated from $n$ multipartite entangled
states, in analogy to  \cite{RGK05,RGK05PRL}:
\begin{equation}
\label{keylength}
\ell^{(n)} = \sup_{{\bf U} \leftarrow {\bf K}}[S_2^\epsilon ({\bf U} E)
-S_0^\epsilon(E) - \max_{i\in\{1,...N-1\}} H_0^\epsilon({\bf U}|{\bf K}_i)]\ ,
\end{equation}
where the smooth R\'enyi entropy $S_\alpha^\epsilon$ is defined as
\begin{equation}
S_{\alpha}^{\epsilon}(\rho)=\frac{1}{1-\alpha} \log_2 (\inf_{\sigma\in \mathbf{B}^\epsilon(\rho)}\tr(\sigma^\alpha)),
\end{equation}
which for $\alpha\in\{0,\infty\}$ is to be understood as $S_\alpha^\epsilon(\rho)=\lim_{\beta\rightarrow\alpha} S_\beta^\epsilon(\rho)$. The infimum is to be taken over all states $\sigma$ in a ball with radius $\varepsilon$ (w.r.t. the trace distance) around $\rho$, denoted as $\mathbf{B}^\epsilon(\rho)$.
For a (classical) probability distribution $P$ the smooth R\'enyi entropy is
\begin{equation}
H_\alpha^\epsilon(P)=\frac{1}{1-\alpha} \inf_{\substack{Q\\\bar{\delta}(Q,P)\leq \epsilon}} \log_2(\sum_z Q(z)^\alpha),
\end{equation}
where the infimum is taken over all probability distributions $\epsilon$-close to $P$ in the sense of the statistical distance $\bar{\delta}$ (the classical analogon of the trace distance).
The conditional smooth R\'enyi entropy is
\begin{equation}
H_0^\epsilon(P_X|P_Y)=\max_y H_0^\epsilon (P_{X|Y=y}).
\end{equation}
Note that, differently to the bipartite case~\cite{RGK05},
the worst of the $N-1$ channels influences the key length via
the maximal leakage to the eavesdropper in the error correction step, see the last term of Eq.~(\ref{keylength}). In the following the symbols $K$, $K_i$ and $U$ denote the single bit random variables corresponding to the respective bold-face strings.
For the limit $n\rightarrow \infty$ the secret fraction $r$ is given by
\begin{eqnarray}
r_{\infty} &=&\lim_{n \rightarrow \infty}\frac{\ell^{(n)}}{n} \nonumber \\
&=&
\sup_{ U \leftarrow K}\inf_{\sigma_{A\{B_i\}} \in \Gamma}
[S (U|E) - \max_{i\in\{1,...N-1\}} H( U|K_i)]
\ , \nonumber \\
&&
\end{eqnarray}
where $S(U|E)$ is the conditional von Neumann entropy, $H(U|K)$ is the conditional Shannon entropy and $\Gamma$ is the set of all density matrices of Alice and the Bobs which are consistent with the parameter estimation.

\section{Details for the network coding example}\label{app:QNC}
Here we explicitly describe the distribution of the GHZ state in the network of Fig.~\ref{fig:router}. This is a special case of the quantum network coding scheme which some of the authors described in \cite{EKB16b}. Let $\ket{+}=\frac{1}{\sqrt{2}}(\ket{0}+\ket{1})$ and $\ket{-}=\frac{1}{\sqrt{2}}(\ket{0}-\ket{1})$.\\
\begin{enumerate}
	\item Alice produces two qubits $C$ and $A$, each in the state $\ket{+}$. She then applies a controlled-Phase gate $C_Z=\proj{0}\otimes \1 + \proj{1}\otimes Z$ to produce the Bell state
	\begin{equation}
	\ket{\includegraphics[width=6ex]{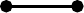}}_{CA}=\frac{1}{\sqrt{2}}(\ket{0+}+\ket{1-}).
	\end{equation}
	\item Alice sends the qubit $C$ to the router station.
	\item The router produces $
	(N-1)$ qubits $B_i$, $i=1,2,...,N-1$, in the state $\ket{+}$ and entangles each of them with the qubit $C$ using $(N-1)$ $C_Z$ gates. At this stage the total state is
	\begin{align}
	\ket{\psi_C}=&\frac{1}{\sqrt{2}}(\ket{0+...+}+\ket{1-...-})\\
	=&\frac{1}{\sqrt{2}}(\ket{+}_C\ket{GHZ'}+\ket{-}_CX_{B_1}\ket{GHZ'}),
	\end{align}
	where
	\begin{align}
	\ket{GHZ'}=&\frac{1}{\sqrt{2}}(\ket{++...+}+\ket{--...-})
	\end{align}
	is the GHZ state in the $X$-basis.
	\item The router measures $C$ in $X$ basis. If the outcome is $-1$, i.e. $\ket{-}_C$, then it applies $X_{B_1}$. The state is now $\ket{\pm}_C\ket{GHZ'}$.
	\item The router now distributes the qubits $B_1$, $B_2$, ..., $B_{N-1}$ to the corresponding parties.
\end{enumerate}
Up to a local basis choice (Hadamard gate), the resource state of the main text has been distributed and the multipartite entanglement based quantum key distribution (NQKD) protocol can be performed.\\
To see that it is impossible to create $N-1$ Bell pairs by sending a single qubit from Alice to the router, let us group the router and all Bobs into a single party $B$. When Alice sends one qubit across the channel, the entropy of entanglement $E_{A|B}\leq 1$. The $N-1$ Bell pairs, however, have entropy of entanglement $E_{A|B}=N-1$, so they cannot be created from the received state by local operations on $B$. Instead, $N-1$ network uses are necessary and the key rate decreases accordingly.
\FloatBarrier

\section{Gate error rates and the QBER}\label{app:gates}
In this appendix we give details for the key rate calculations regarding the quantum networks of Figs. \ref{fig:multipartite} and \ref{fig:router} with imperfect gates. We start with the simple network of Fig.~\ref{fig:multipartite}. The GHZ resource state is prepared as follows. Alice starts with the state $\ket{+}_A\ket{0}^{\otimes N-1}$ and applies a controlled-Not gate from $A$ to each of the other qubits, see Fig.~\ref{fig:circuit}. %
\begin{figure}[tbp]%
	\centering%
	\includegraphics[scale=1.0]{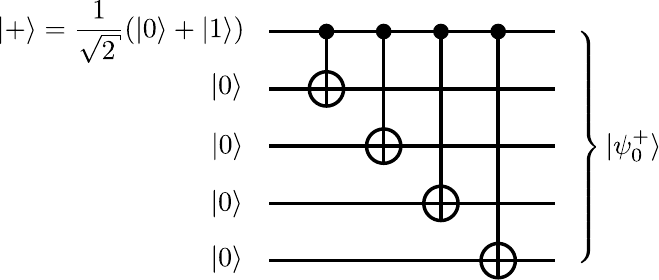}%
	\caption{The GHZ state that is to be distributed across the network of Fig.~\ref{fig:multipartite} can be produced by Alice using controlled-Not gates as depicted in this quantum circuit diagram.}\label{fig:circuit}%
\end{figure}%
When a controlled-Not gate acts on qubits $i$ (control) and $j$ (target) we denote it by
\begin{equation}
C_X^{(i,j)}=(\proj{0}_i \otimes \1_j+\proj{1}_i \otimes X_j)\otimes \1_{\mathrm{rest}},
\end{equation}
where $X=\ket{0}\bra{1}+\ket{1}\bra{0}$ is a Pauli matrix. We use a depolarizing noise model for the gate errors. The action of the imperfect gate on the density matrix is
\begin{align}
C_{X,f_G}^{(i,j)}(\rho)=& (1-f_G) C_X^{(i,j)} \rho C_X^{(i,j)} + f_G \tr_{i,j}(\rho) \otimes \1_{ij}\\
=&(1-f_G) C_X^{(i,j)} \rho C_X^{(i,j)} + \sum_{a,b\in\sigma} a_i b_j \rho a_i b_j,
\end{align}
where $\sigma=\{\1,X,Y,Z\}$ contains Pauli matrices.\\
It will be convenient to extend the notation of the GHZ basis to include the number of parties as a subscript, i.e.
\begin{equation}
\ket{\psi_{j,N}^\pm} = \frac{1}{\sqrt{2}}(\ket{0}_{1}\ket{j}_{2...N}\pm \ket{1}_{1}\ket{\bar j}_{2...N}). \label{eq:ghz-basisext}
\end{equation}
The initial state is
%\begin{widetext}
\begin{equation}
\rho_{\mathrm{in}} = \proj{\psi_{0,1}^+} \otimes (\proj{0})^{\otimes(N-1)}.
\end{equation}
The first gate turns it into
\begin{equation}
\rho_{1,\mathrm{out}} = ((1-f_G) \proj{\psi_{0,2}^+}+f_G \frac{\1}{4})\otimes (\proj{0})^{\otimes N-2},
\end{equation}
the second into
\begin{align}
\rho_{2,\mathrm{out}} = &((1-f_G)^2 \proj{\psi_{0,3}^+}\nonumber \\
+&(1-f_G) f_G \frac{1}{2} (\proj{0}\otimes\frac{\1}{2}\otimes\proj{0}+\proj{1}\otimes\frac{\1}{2}\otimes\proj{1})\nonumber\\
+&f_G \frac{\1}{8})\otimes (\proj{0})^{\otimes N-3},
\end{align}
and the third into
\begin{align}
\rho_{3,\mathrm{out}}=&((1-f_G)^3\proj{\psi_{0,4}^+}\nonumber\\
&+f_G(1-f_G)^2 \frac{1}{2}(\frac{\1}{2}\otimes \proj{00} \otimes \frac{\1}{2}+\frac{\1}{2}\otimes \proj{11} \otimes \frac{\1}{2})\nonumber\\
&+(1-f_G)^2 f_G \frac{1}{2} (\proj{0}\otimes\frac{\1}{2}\otimes\proj{00}+\proj{1}\otimes\frac{\1}{2}\otimes\proj{11})\nonumber\\
&+(1-f_G) f_G \frac{1}{2}(\proj{0}\otimes \frac{\1}{4} \otimes \proj{0}+\proj{1}\otimes \frac{\1}{4} \otimes \proj{1})\nonumber\\
&+(2-f_G)f_G^2 \frac{\1}{16}) \otimes (\proj{0})^{\otimes N-4}.
\end{align}
One may deduce the following observation. Let us denote the pattern of actual gate successes/failures as the binary representation of an $(N-1)$-bit number $\mathbf{x}$, where a $0$ at position $i$ indicates the failure of gate $i$ and $1$ means the corresponding gate was successful. The number of connected blocks of ones in the bit string $\mathbf{x}1$ plus the number of zeros, $b(\mathbf{x}1)$, is the number of subsets of parties that are correlated amongst each other. This gives the prefactor
\begin{equation}
c_\mathbf{x}=\left\{\begin{array}{cl}
1 & \text{if } \mathbf{x}=11...1\\
2^{-b(\mathbf{x}1)} &\text{else } \end{array}\right.
\end{equation}
in front of the corresponding term in $\rho$. These prefactors determine the overlap between $\ket{\psi_{j,0}^{\pm}}\bra{\psi_{j,0}^{\pm}}$ and $\rho$, i.e. the coefficients $\lambda_0^\pm(f_G)$ of $\rho$ in the GHZ basis. They read
\begin{equation}
\begin{aligned}
\lambda_0^+(f_G)=&\sum_{\mathbf{x}=0}^{2^{N-1}-1} c_\mathbf{x} f_G^{N-1-|\mathbf{x}|_H} (1-f_G)^{|\mathbf{x}|_H}\\
\text{and\hspace{0.8cm}}\lambda_0^-(f_G)=&\lambda_0^+-(1-f_G)^{N-1},
\end{aligned}\label{eq:applambda0fG}
\end{equation}
where $|\mathbf{x}|_H$ is the Hamming weight of $\mathbf{x}$. After some combinatorics $\sum_{\mathbf{x}} c(\mathbf{x})$ for a given weight $w=|\mathbf{x}|_H$ (unequal to $N-1$) can be expressed in a more compact form by summing over all possible ``subset counts'' $\beta$ as
\begin{equation}
\sum_{\substack{\mathbf{x}\\|\mathbf{x}|_H=w}} c(\mathbf{x})= \sum_{\beta=N-w}^{N}{w \choose N-\beta} {N-w-1 \choose \beta -N+w} 2^{-\beta},
\end{equation}
which leads to the relevant coefficients in the GHZ basis,
\begin{equation}
\begin{aligned}
\lambda_0^-(f_G)=&\sum_{w=0}^{N-2} c'(w) f_G^{N-1-w} (1-f_G)^{w}\\
\text{and }\lambda_0^+(f_G)=&(1-f_G)^{N-1}+\lambda_0^{-}(f_G)
\end{aligned}\label{eq:lambda0fG}
\end{equation}
with
\begin{equation}
c'(w)=\sum_{\beta=N-w}^{N}{w\choose N - \beta} {N - w - 1\choose \beta - N + w} 2^{-\beta}.
\end{equation}
From Eq.~(\ref{eq:lambda0fG}) one obtains the QBER using Eq.~(\ref{eq:QBER}). We show it in Fig.~\ref{fig:QoffG}.%
\begin{figure}[htbp]
	\centering \includegraphics[scale=0.9]{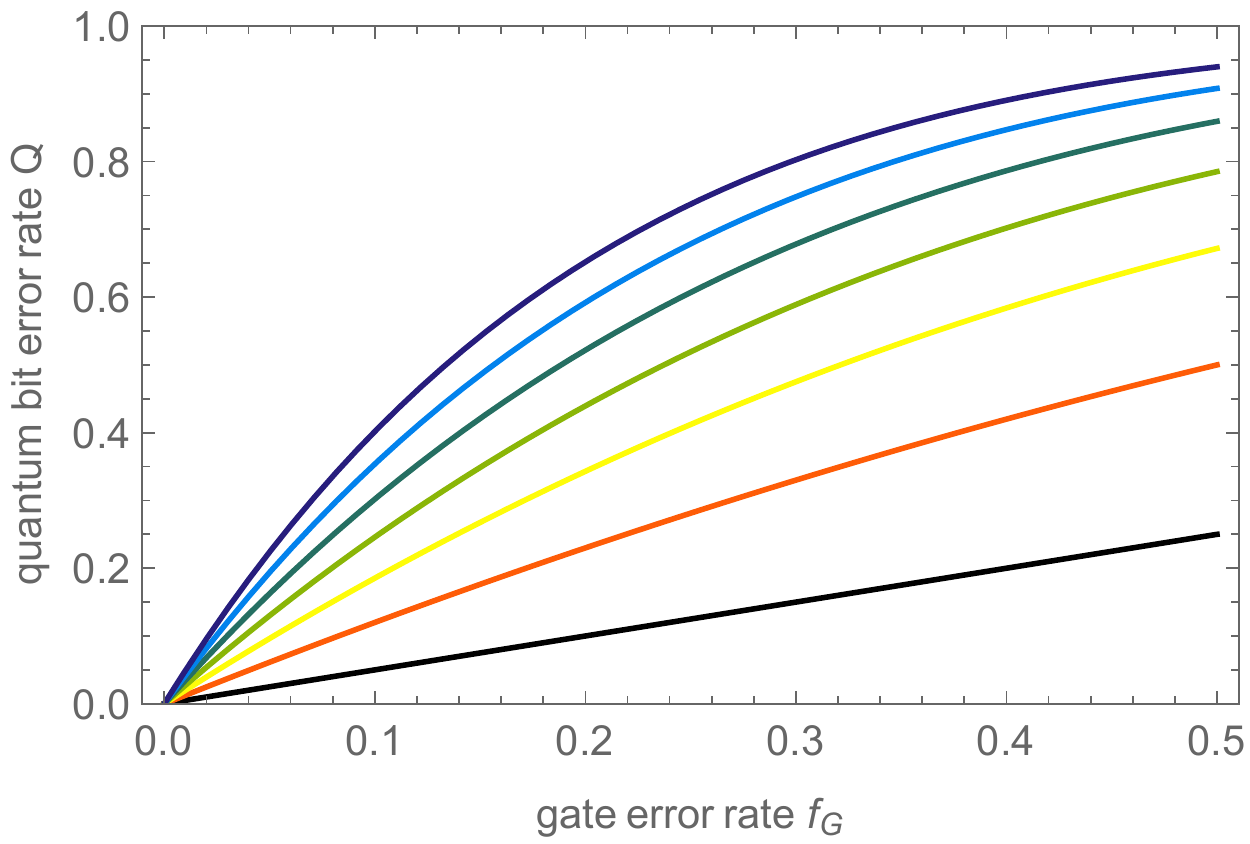}
	\caption{The QBER as a function of $f_G$ for the circuits described in the text, with $N=2,3,4,...,8$ (bottom to top).}\label{fig:QoffG}
\end{figure}
The secret key rate is calculated using Eq.~(\ref{eq:rdep}) with
\begin{align}
Q_{AB_i} =& \frac{1}{N-1}\sum_{k=1}^{N-1} \frac{1}{2}(1-(1- f_G)^k)\\
=& \frac{(1 - f_G)^N + f_G N-1}{2 f_G (N-1)},
\end{align}
which is the average $Q_{AB_i}$ for one to $N-1$ gates, because we use a random order of the gates. This effectively mixes all $\lambda_j^\pm$ with $j$ of same Hamming weight and accomplishes that all $Q_{AB_i}$ are equal. Compared to a fixed gate order it improves the key rate and removes the maximum in Eq.~(\ref{eq:rdep}).\\
In the case of the network shown in Fig.~\ref{fig:router}, $N-1$ gates are performed at $C$ and one additional gate is performed at $A$. The initial state at $C$ depends on whether the gate of $A$ was successful, i.e. it is
\begin{align}
\rho_{\mathrm{in,QNC}}=&(1-f_G) \rho_{\mathrm{in}} + f_G \frac{\1}{2}\otimes \proj{0}^{\otimes (N-1)}\\
=& (1-f_G) \rho_{\mathrm{in}} + f_G \frac{1}{2} (\rho_{\mathrm{in}}+Z_1 \rho_{\mathrm{in}} Z_1),
\end{align}
i.e.
\begin{align}
\lambda_{0,QNC}^+(f_G)=&(1-f_G) \lambda_{0}^+(f_G) + \frac{f_G}{2}(\lambda_0^+(f_G)+\lambda_0^-(f_G))\\
\lambda_{0,QNC}^-(f_G)=&(1-f_G) \lambda_{0}^-(f_G) + \frac{f_G}{2}(\lambda_0^+(f_G)+\lambda_0^-(f_G)),
\end{align}
and the previous results can be used to obtain the key rate in this case. Note that while the final density matrix depends on whether a router was used or not, the QBER (and $Q_{AB_i}$) does not, because the additional phase error does not contribute to it.\\
For a fixed number of parties $N$ there is a threshold gate error probability below which the NQKD protocol outperforms the bipartite approach in the quantum network of Fig.~\ref{fig:router}. These values are listed in Table~\ref{tab:minfG}. \begin{table}[tbp]%
	\caption{The multipartite entanglement based QKD protocol is more prone to gate errors, but requires Alice to send one qubit only. These two competing effects lead to a threshold value of the gate error probability $f_G$ below which it outperforms the bipartite approach.}\label{tab:minfG}%
	\begin{tabular}{cc}%
		$N$ & NQKD-threshold for $f_G$\\
		\hline%
		3 & 0.0725754 \\%
		4 & 0.0689939 \\%
		5 & 0.0618163 \\%
		6 & 0.0553032 \\%
		7 & 0.0498258 \\%
		8 & 0.0452567 \\%
		9 & 0.0414201 \\%
		10 & 0.0381659 \\%
	\end{tabular}\hspace{1cm}%
	\begin{tabular}{cc}%
		$N$ & NQKD-threshold for $f_G$\\
		\hline%
		11 & 0.0353766 \\%
		12 & 0.0329621 \\%
		13 & 0.0308531 \\%
		14 & 0.0289959 \\%
		15 & 0.0273484 \\%
		16 & 0.0258773 \\%
		17 & 0.024556 \\%
		18 & 0.0233626 \\%
	\end{tabular}%
\end{table}%

The key rate as a function of $N$ is shown for different values of the gate error rate $f_G$ in Fig.~\ref{fig:keyratesQNC}.
\begin{figure}[tb]%
	\centering%
	\includegraphics[scale=0.9]{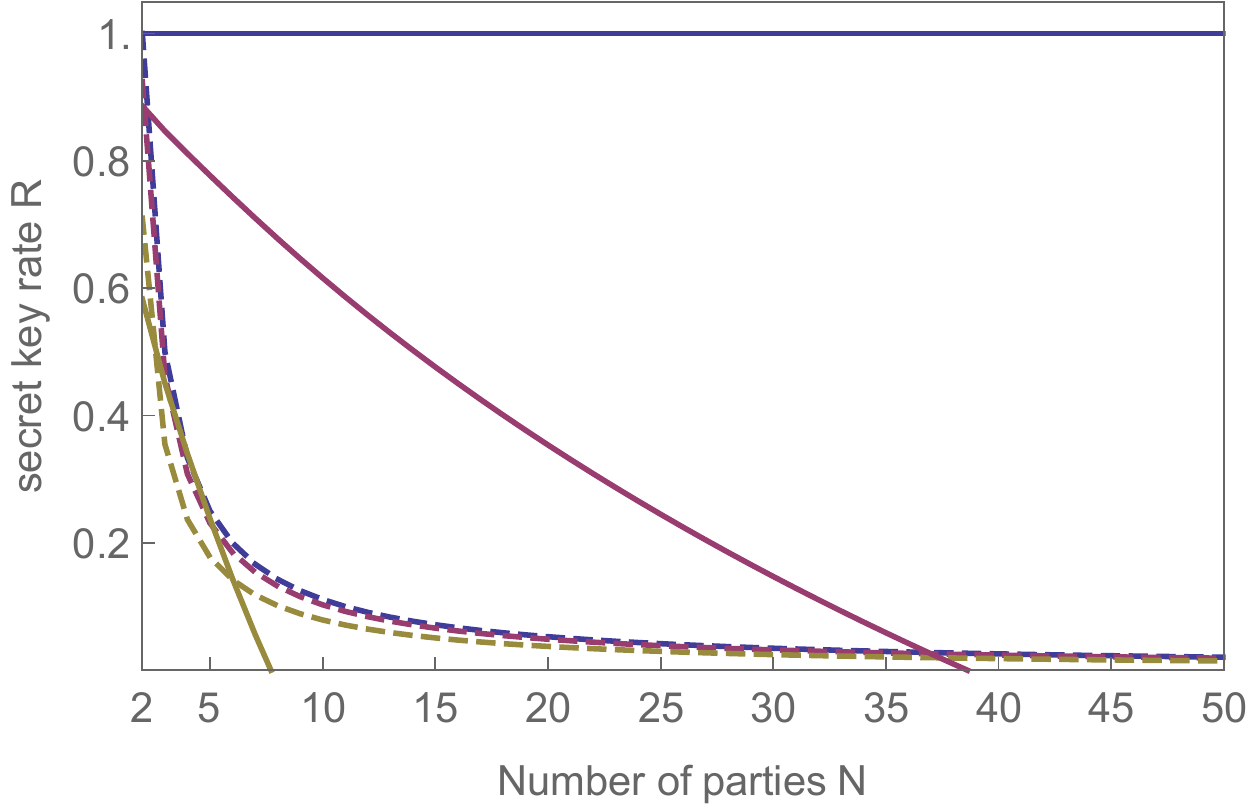}%
	\caption{The secret key rate for the multipartite entanglement based (NQKD) protocol (solid lines) in the quantum network shown in Fig.~\ref{fig:router} as a function of the number of parties $N$ for different values of the gate error probability $f_G=0\,\%,\, 1\,\% \text{ and }\,5\,\%$ (top to bottom). The key rate decreases with increasing $N$, because more imperfect gates are applied. The key rate of the bipartite entanglement based (2QKD) protocol (dashed lines), plotted for the same values of $f_G$, shows the $1/(N-1)$ scaling which is due to the bottleneck between $A$ and $C$.
	}
	\label{fig:keyratesQNC}%
\end{figure}%
\FloatBarrier
\section{Key distribution in the butterfly network}\label{app:butterfly}
\begin{figure}[htbp]
	\centering
	\subfigure[A (classical) linear network code\label{fig:butterflya}]{\hspace{1cm}\includegraphics[scale=1.2]{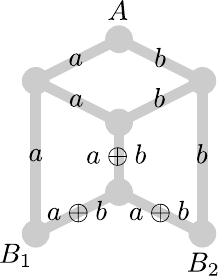}\hspace{1cm}}
	\subfigure[The corresponding quantum network code produces two GHZ states.\label{fig:butterflyb}] {\hspace{1cm}\includegraphics[scale=1.2]{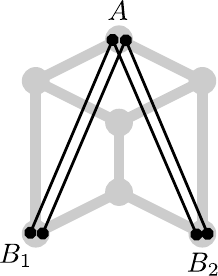}\hspace{1cm}}
	\caption[Sketch]{A butterfly network. In the classical case (a) A can send bits $a$ and $b$ to both $B_1$ and $B_2$ by employing the linear network code given by the transmitted symbols written onto the channels ($\oplus$ means XOR). This implies that the corresponding quantum network code produces two GHZ states (b), see Thm. 1 of \cite{EKB16b}.}
	\label{fig:butterfly}
\end{figure}
We sketch how the NQKD protocol can be employed in the butterfly network shown in Fig.~\ref{fig:butterfly}. As usual, the rate constraints on the channels are one, i.e. each channel can send a single qubit per time step.
\begin{enumerate}
	\item The quantum network code corresponding to the linear code shown in FIG.~S\ref{fig:butterflya} is employed to produce two GHZ states shared by $A$, $B_1$ and $B_2$ (FIG.~S\ref{fig:butterflyb}). See Thm. 1 of \cite{EKB16b}.
	\item These two GHZ states allow to perform two rounds of the NQKD protocol in a single time step.
\end{enumerate}
In contrast the bipartite entanglement based (2QKD) protocol (also in its prepare and measure formulation) can only do a single round, because only two Bell pairs can be distributed (due to the outgoing capacity at A). Thus the key rate of the NQKD protocol is twice as high as in the ``standard approach''.\\
From the construction of this example it is clear how it generalizes: If the network allows A to multicast $n$ bits, then a single use of the corresponding quantum network will produce $n$ GHZ states. Thus the NQKD protocol can be performed $n$ times per time step. However, the 2QKD protocol can only perform $\frac{n}{N-1}$ rounds in the same time.

\end{document}